\documentclass[pra,longbibliography,twocolumn,showpacs,nofootinbib,superscriptaddress,notitlepage]{revtex4-2}
\usepackage{amsmath}
\usepackage{amssymb,bm}
\usepackage{amsthm}

\usepackage{qcircuit}

\newtheorem{innercustomgeneric}{\customgenericname}
\providecommand{\customgenericname}{}
\newcommand{\newcustomtheorem}[2]{%
  \newenvironment{#1}[1]
  {%
   \renewcommand\customgenericname{#2}%
   \renewcommand\theinnercustomgeneric{##1}%
   \innercustomgeneric
  }
  {\endinnercustomgeneric}
}

\newcustomtheorem{customthm}{Theorem}
\newcustomtheorem{definitionBob}{Definition}

\usepackage{color,dsfont} 
\usepackage{graphicx}
\usepackage{subcaption}
\usepackage{ragged2e}
\DeclareCaptionJustification{justified}{\justifying}
\usepackage[colorlinks=true, hyperindex, breaklinks, linkcolor=blue, urlcolor=blue, citecolor=blue]{hyperref} 
\usepackage[normalem]{ulem}
\usepackage[capitalise]{cleveref}
\usepackage{mathrsfs}

\usepackage{mathtools}
\usepackage{float}
\usepackage{verbatim}
\usepackage{latexsym}
\usepackage{amsmath}
\usepackage{amssymb}
\usepackage{setspace}
\usepackage{amsfonts}
\usepackage{stmaryrd}
\usepackage{xcolor}
\usepackage{enumitem}
\usepackage{hhline}
\usepackage[normalem]{ulem}

\DeclareMathOperator{\wt}{wt}

\newtheorem{theorem}{Theorem} \newtheorem{lemma}{Lemma}
\newtheorem{corollary}{Corollary} \newtheorem{conjecture}{Conjecture}
 
\newtheorem{proposition}{Proposition}

\newcommand{\bra}[1]{\langle#1|}  
\newcommand{\ket}[1]{|#1\rangle}  
 
\crefname{equation}{Eq.\!}{Eqs.\!}
\crefname{figure}{Fig.\!}{Figs.\!}
\usepackage{multirow}
\mathchardef\mhyphen="2D

\usepackage{tikz}
\usepackage{pgfplots}
\pgfplotsset{every axis legend/.append style={at={(1,0)},anchor=south east}}

\begin{document}

\title{A Hamming-Like Bound for Degenerate Stabilizer Codes}

\author{Andrew Nemec}
\email{andrew.nemec@duke.edu}
\affiliation{
	Duke Quantum Center, Duke University, Durham, NC 27701, USA
}
\affiliation{
	Department of Electrical and Computer Engineering, Duke University, Durham, NC 27708, USA
}
\author{Theerapat Tansuwannont}
\email{t.tansuwannont@duke.edu}
\affiliation{
	Duke Quantum Center, Duke University, Durham, NC 27701, USA
}
\affiliation{
	Department of Electrical and Computer Engineering, Duke University, Durham, NC 27708, USA
}

\begin{abstract}
The quantum Hamming bound was originally put forward as an upper bound on the parameters of nondegenerate quantum codes, but over the past few decades much work has been done to show that many degenerate quantum codes must also obey this bound. In this paper, we show that there is a Hamming-like bound stricter than the quantum Hamming bound that applies to degenerate $t$-error-correcting stabilizer codes of length greater than some positive integer $N(t)$. We show that this bound holds for all single-error-correcting degenerate stabilizer codes, forcing all but a handful of optimal distance-3 stabilizer codes to be nondegenerate.
\end{abstract}


\maketitle

\section{Introduction}

Quantum error correction and fault tolerance are necessary for the construction of large-scale quantum computers. Much of the theory of quantum error correction is adopted from classical error correction, including the well known class of quantum stabilizer codes \cite{Gottesman1997}, which has a direct correlation with the class of classical self-orthogonal additive codes \cite{Calderbank1998}. However, quantum codes may have attributes not found among classical codes, such as degeneracy, in which multiple correctable errors may have identical effects on the encoded state.

Degenerate codes are important to the theory of quantum error correction, as they make up a large portion of the class of quantum LDPC codes \cite{Tillich2014, Panteleev2022, Leverrier2022}, as well being important for constructions of quantum-classical hybrid codes \cite{Grassl2017, Nemec2021} and permutation-invariant quantum codes \cite{Pollatsek2004, Kubischta2023}. Regardless of their importance, degenerate codes remain understudied, mostly on account of the difficulty of analyzing them, especially when proving bounds on their parameters.

One of the simplest bounds on the parameters of classical codes is the Hamming (or sphere-packing) bound \cite{Hamming1950}. However, with the quantum Hamming bound \cite{Ekert1996, Knill1997}, the analogous argument only work for nondegenerate codes. While initially there was hope that degenerate codes might surpass the quantum Hamming bound by ``packing'' more efficiently \cite[Section 10.3.3]{Nielsen2000}, through a series of results \cite{Gottesman1997, Li2010, Dallas2022} the quantum Hamming bound has been shown to hold for many degenerate codes, leading to the conjecture that the quantum Hamming bound must hold for all quantum codes.

Recently, it has been shown that for single-error-correcting degenerate stabilizer codes, a bound slightly stronger than the quantum Hamming bound holds \cite{Nemec2023}. In this paper, we strengthen this argument and show that there exists a positive integer $N\!\left(t\right)$ such that for all $n\geq N\!\left(t\right)$, any $t$-error-correcting degenerate stabilizer code of length $n$ must satisfy
\begin{equation}
k\leq n-1-\log_{2}\!\left(\sum\limits_{i=0}^{t}3^{i}\binom{n-2t}{i}\right),
\end{equation}
which is tighter than the quantum Hamming bound.

\section{Degenerate Codes}

An $\left(\!\left(n,K,d\right)\!\right)$ quantum code is a $K$-dimensional subspace $\mathcal{C}$ of the Hilbert space $\mathbb{C}^{2^{n}}$ defined by $n$ physical qubits, such that any error affecting $t=\left\lfloor\frac{d-1}{2}\right\rfloor$ physical qubits can be corrected. The set of correctable errors $\mathcal{E}$ is defined by the Knill-Laflamme conditions \cite{Knill1997}:
\begin{equation}
\bra{c_{i}}E_{a}^{\dagger}E_{b}\ket{j}=\lambda_{a,b}\delta_{i,j},
\end{equation}
for all $E_{a},E_{b}\in\mathcal{E}$, where $\left\{\ket{c_{i}}\right\}$ is an orthonormal basis of $\mathcal{C}$. If all errors affecting $t$ or fewer qubits are correctable, we call $\mathcal{C}$ a $t$-error-correcting code.

The coefficients $\lambda_{a,b}$ define a Hermitian matrix, and we call the quantum code $\mathcal{C}$ degenerate with respect to $\mathcal{E}$ if and only if this associated matrix is degenerate. In other words, a nondegenerate code will have linearly independent correctable errors that map the code to linearly independent subspaces. A related dichotomy is the distinction between pure and impure codes \cite{Calderbank1998}, where a code is called pure with respect to $\mathcal{E}$ if linearly independent errors map the code to mutually orthogonal subspaces. In particular, code purity implies nondegeneracy, although the converse does not need to hold (see \cite{Cao2022} for an example of an impure but nondegenerate nonadditive code).

In this paper we will restrict our focus to the class of stabilizer codes, where there is no distinction between impure and degenerate codes. In this case, an $\left[\!\left[n,k,d\right]\!\right]$ stabilizer code, where $K=2^{k}$, is degenerate if and only if its stabilizer group $\mathcal{S}$ contains a non-identity element of weight less than $d$. One property of degenerate stabilizer codes is that they have low-weight correctable Pauli errors that share syndromes.

\section{Quantum Hamming Bound}

For a classical $\left(n,K,d\right)$ error-correcting code over an alphabet of cardinality $q$, each codeword must be at least Hamming distance $d$ away from every other codeword, meaning that Hamming-spheres of radius $t$ centered around codewords cannot overlap. This gives rise to the classical Hamming bound:
\begin{equation}
\log_{q}\!\left(\sum\limits_{i=0}^{t}\left(q-1\right)^{i}\binom{n}{i}\right)\leq n-\log_{q}\!\left(K\right).
\end{equation}
This is also known as the sphere-packing bound, as at most $K$ spheres of radius $t$ can fit in the Hamming space defined by codewords of length $n$.

A quantum analog of this bound was introduced independently by Ekert and Macchiavello \cite{Ekert1996} and Knill and Laflamme \cite{Knill1997} for nondegenerate $\left(\!\left(n,K,d\right)\!\right)$ codes:
\begin{equation}
\log_{2}\!\left(\sum\limits_{i=0}^{t}3^{i}\binom{n}{i}\right)\leq n-\log_{2}\!\left(K\right).
\end{equation}
From here forward, we restrict ourselves to the stabilizer case. In this case $K=2^{k}$, and we can interpret the left hand side of the inequality as the logarithm of the number of possible errors of weight at most $t$ on $n$ qubits, and the right hand side as the logarithm of the number of possible syndromes. For a nondegenerate code, each of these low-weight errors must have a unique syndrome, so it must obey this bound, while for degenerate codes at least one syndrome is shared by multiple errors, causing this argument to break down.

Since this sphere-packing argument seemingly does not apply to degenerate codes, it has been suggested that these codes might be able to violate the quantum Hamming bound by some efficient packing of small-weight errors into shared syndromes \cite[Section 10.3.3]{Nielsen2000}. This has been supported by the discovery of examples of highly-degenerate codes which have been found to violate the quantum hashing bound \cite{DiVincenzo1998, Kay2014}.

However, even early on this view was challenged by results by Gottesman that all stabilizer codes with $t=1,2$ must satisfy the quantum Hamming bound \cite{Gottesman1997} and by Rains that it holds for all codes with $n\leq 30$ (see \cite{Shor1997}). Li and Xing later gave more supporting evidence by showing there exists a positive integer $M\!\left(t\right)$ such that for all $n\geq M\!\left(t\right)$, any degenerate $\left(\!\left(n,K,2t+1\right)\!\right)$ code must satisfy the quantum Hamming bound \cite{Li2010}. More recently, Dallas et al. combined the result of Li and Xing with a bound due to Rains \cite{Rains1999b} to show that the bound holds for all codes with $t<63$ \cite{Dallas2022}.

This evidence leads to the following conjecture:

\begin{conjecture}
All quantum codes satisfy the quantum Hamming bound.
\end{conjecture}

An asymptotic version of the quantum Hamming bound was proven by Ashikhmin et al. \cite{Ashikhmin1999, Ashikhmin2000}. Generalizations of the bound to the nonbinary case were done in \cite{Feng2004, Ketkar2006, Aly2007, Sarvepalli2010}. A bound combining the quantum Hamming and Singleton bounds was also given by Yu et al. \cite{Yu2010}.

There have been many variants of the quantum Hamming bound that have been proven for different variants of quantum codes, including subsystem codes \cite{Aly2006, Klappenecker2007a, Klappenecker2007b}, entanglement-assisted codes \cite{Lai2013, Li2014, Lai2018}, and quantum data-syndrome codes \cite{Fujiwara2014, Ashikhmin2020, Nemec2023}. In the former two cases, examples of codes violating those variants of the quantum Hamming bound have been found.

\section{New Bounds for Degenerate Stabilizer Codes}

Taking inspiration from the quantum data-syndrome variant of the quantum Hamming bound, recently it has been proven in \cite{Nemec2023} that any degenerate $\left[\!\left[n,k,3\right]\!\right]$ stabilizer code must satisfy the bound
\begin{equation}
\log_{2}\!\left(4n-k+1\right)\leq n-k,
\end{equation}
which is stronger than the quantum Hamming bound. We believe that this is the first bound that applies to degenerate codes but not to nondegenerate codes, with the exception of the quantum Singleton bound, which has a non-strict inequality in the case of nondegenerate codes and a strict inequality in the case of degenerate codes \cite{Calderbank1998, Rains1999a}.

Can this bound be improved, and if so, by how much can it be improved? Clearly it cannot be dramatically far from any bound for nondegenerate codes, because the existence of an $\left[\!\left[n,k,d\right]\!\right]$ stabilizer code with $k>0$ always implies the existence of an $\left[\!\left[n+1,k,d\right]\!\right]$ degenerate stabilizer code \cite[Theorem 6]{Calderbank1998}. In \cite{Nemec2023}, it was also conjectured that a stronger bound
\begin{equation}
\log_{2}\!\left(3\!\left(n-2\right)+1\right)+1\leq n-k,
\end{equation}
also holds for single-error-correcting stabilizer codes. This conjectured bound is quite tight with many of the best known degenerate stabilizer codes constructed using \cite[Theorem 6]{Calderbank1998}. In this section we prove that this conjectured bound holds for all $\left[\!\left[n,k,3\right]\!\right]$ stabilizer codes, and prove that there exists a positive integer $N\!\left(t\right)$ such that for all $n\geq N\!\left(t\right)$, any $\left[\!\left[n,k,2t+1\right]\!\right]$ degenerate stabilizer code must satisfy a similar bound.

We go about proving these new bounds using the same technique used by Gottesman in \cite{Gottesman1997} to prove the quantum Hamming bound for degenerate stabilizer codes with $t=1,2$.

\begin{lemma}
Any degenerate $\left[\!\left[n,k,2t+1\right]\!\right]$ stabilizer code with $\ell$ independent stabilizer generators $S_{1}, \dots, S_{\ell}$ of weight $2t$ or less satisfies
\begin{equation}
k\leq \begin{cases}
n-\ell-\log_{2}\!\left(\sum\limits_{i=0}^{t}3^{i}\binom{n-\sigma}{i}\right) & \text{if } \sigma\leq n \\
n-\ell & \text{if } \sigma> n,
\end{cases}
\label{gensigmabound}
\end{equation}
where $\sigma=\sum_{i=1}^{\ell}\wt\!\left(S_{i}\right)$.
\end{lemma}
\begin{proof}
Let $\mathcal{S}$ be the stabilizer group of the code $\mathcal{C}$, and let $\mathcal{S}_{D}=\left\langle S_{1},\dots,S_{\ell}\right\rangle$ be the subgroup generated by the stabilizer generators of weight less than $2t$. Then we have that $\mathcal{S}\setminus \mathcal{S}_{D}$ has no elements of weight less than $2t$. The subspace $\mathcal{C}_{D}$ stabilized by $\mathcal{S}_{D}$ has dimension $2^{n-\ell}$, and an operator in the normalizer of $\mathcal{S}_{D}$ will map states in $\mathcal{C}_{D}$ back to $\mathcal{C}_{D}$. Suppose that the $j$-th qubit is not acted upon by the operators in $\mathcal{S}_{D}$, then the single-qubit Pauli operators $X_{j}$, $Y_{j}$, and $Z_{j}$ must all be in $N\!\left(\mathcal{S}_{D}\right)$ and must not be degenerate errors. Therefore, each must map a basis codeword of $\mathcal{C}$ to an orthogonal codeword in $\mathcal{C}_{D}$. Since $\sigma$ is an upper bound on the number of physical qubits affected by the operators in $\mathcal{S}_{D}$, there are always at least $n-\sigma$ unaffected qubits. Therefore, the following holds when $\sigma\leq n$:
\begin{equation}
2^{k}\sum\limits_{i=0}^{t}3^{i}\binom{n-\sigma}{i}\leq 2^{n-\ell}.
\end{equation}

When $\sigma>n$, we can still bound $k$ by noting that $\mathcal{S}$ is generated by $n-k$ elements, therefore $\ell\leq n-k$, which can be rearranged as
\begin{equation}
k\leq n-\ell.
\end{equation}
\end{proof}

\begin{figure}[t]
\scalebox{0.88}{
\begin{tikzpicture}
\begin{axis}[
samples = 50,
xlabel = {$n$},
xmajorgrids=true,
xmax = 26.1,
xmin = 0.0,
y label style={at={(axis description cs:0.1,.5)},anchor=south},
ylabel = {$k$},
ymajorgrids=true,
ymax = 18.1,
ymin = 0.0,
domain = 0.0:26.1
]
\addplot[solid, black, ultra thick] {x - (ln(1 + 3 * x) / ln(2))};
\addplot[solid, red, ultra thick][domain=6.0:26.1] {x - 3 - (ln(1 + 3 * (x - 6)) / ln(2))};
\addplot[solid, red, ultra thick][domain=0.0:6.0] {x - 3};
\addplot[solid, blue, ultra thick][domain=14.0:26.1] {x - 8 - (ln(1 + 3 * (x - 14)) / ln(2))};
\addplot[solid, blue, ultra thick][domain=0.0:14.0] {x - 8};

\end{axis}
\end{tikzpicture}
}
\caption{The quantum Hamming bound ($\ell=0,\sigma=0$) for $t=1$ (black) and two of its shifts: $\ell=3,\sigma=6$ (red) and $\ell=8,\sigma=14$ (blue).}
\label{shifthammingfig}
\end{figure}
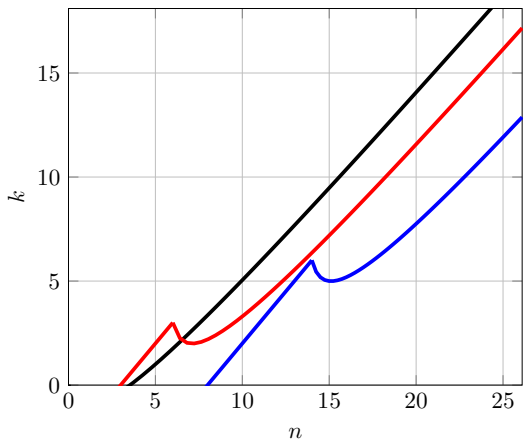

It is easy to see that the portion of the bound in \cref{gensigmabound} where $\sigma\leq n$ is simply a copy of the quantum Hamming bound that has been horizontally shifted to the right by $\sigma$ and vertically shifted up by $\sigma-\ell$:
\begin{equation}
k\leq \begin{cases}
h_{t}\!\left(n-\sigma\right)+\sigma-\ell & \text{if } \sigma\leq n \\
n-\ell & \text{if } \sigma> n,
\end{cases}
\label{shiftsigmabound}
\end{equation}
where $h_{t}\!\left(n\right)$ is the usual quantum Hamming bound for $t$-error-correcting codes. \cref{shifthammingfig} shows the quantum Hamming bound for $t=1$ and two of its shifts. We will need multiple results about the function $h_{t}\!\left(n\right)$ and its shifts, which we prove in the Appendix.


If $\ell$ is held constant, increasing $\sigma$ by one results in a translation of \cref{shiftsigmabound} up by one and to the right by one. By \cref{shift1-1} the Appendix, we know that $\sigma=2t\ell$ also results in a bound for $k$ for all values of $\sigma$ in the range $\ell\leq\sigma\leq2t\ell$. We refer to 
\begin{equation}
k\leq \begin{cases}
n-\ell-\log_{2}\!\left(\sum\limits_{i=0}^{t}3^{i}\binom{n-2t\ell}{i}\right) & \text{if } 2t\ell\leq n \\
n-\ell & \text{if } 2t\ell> n,
\end{cases}
\end{equation}
as the $\left(\ell,t\right)$-bound, as it is an upper bound on any degenerate $\left[\!\left[n,k,2t+1\right]\!\right]$ stabilizer code with $\ell$ independent generators. Note that the $\left(0,t\right)$-bound is the original quantum Hamming bound. We show that for any fixed values of $\ell$ and $t$, there exists some positive integer $N\!\left(\ell,t\right)$ after which the $\left(\ell,t\right)$-bound will dominate the $\left(\ell',t\right)$-bound for all $\ell'>\ell$.

\begin{theorem}
There exists an $N\!\left(\ell,t\right)\in\mathbb{N}$ such that for all $n\geq N\!\left(\ell,t\right)$, any $\left[\!\left[n,k,2t+1\right]\!\right]$ stabilizer code with at least $\ell$ independent generators of weight $2t$ or less satisfies the following inequality:
\begin{equation}
k\leq n-\ell-\log_{2}\!\left(\sum\limits_{i=0}^{t}3^{i}\binom{n-2t\ell}{i}\right).
\end{equation}
\end{theorem}
\begin{proof}
We want to show that there exists some $N\!\left(\ell,t\right)\in\mathbb{N}$ after which the $\left(\ell,t\right)$-bound will dominate all $\left(\ell',t\right)$-bounds for $\ell'>\ell$. Note that the $\left(\ell,t\right)$-bound is the $\left(0,t\right)$-bound shifted to the right by $2t\ell$ and up by $\left(2t-1\right)\ell$, and similarly, the $\left(\ell',t\right)$-bounds are shifts of the $\left(a,t\right)$-bounds by the same amount, where $a=\ell'-\ell$. Therefore, $N\!\left(\ell,t\right)$ exists and equals $N\!\left(0,t\right)+2t\ell$ if and only if $N\!\left(0,t\right)$ exists.

The $\left(a,t\right)$-bound has a local maxima at $n=2ta$, which will be less than the $\left(0,t\right)$-bound at that point for all $a$ such that 
\begin{equation}
\log_2\left(\sum\limits_{i=0}^{t}3^{i}\binom{2ta}{i}\right)< a.
\end{equation}
Since the left hand side of the inequality is in $o\!\left(a\right)$, it is true that this must hold true for all $a$ after some $a_{0}$. By \cref{shiftprop} the Appendix, all $\left(a,t\right)$-bounds for $a\geq a_{0}$ must be bounded above by the $\left(0,t\right)$-bound for all values of $n$.

For $a<a_{0}$, the difference between the $\left(0,t\right)$-bound and the $\left(a,t\right)$-bound is
\begin{equation}
n-\log_{2}\!\left(\sum\limits_{i=0}^{t}3^{i}\binom{n}{i}\right)-n+a+\log_{2}\!\left(\sum\limits_{i=0}^{t}3^{i}\binom{n-2ta}{i}\right),
\end{equation}
or equivalently
\begin{equation}
a+\log_{2}\!\left(\frac{\sum\limits_{i=0}^{t}3^{i}\binom{n-2ta}{i}}{\sum\limits_{i=0}^{t}3^{i}\binom{n}{i}}\right),
\end{equation}
which for large $n$ approaches $a$. Since $a>0$, it follows that there exists some value $n_{a}$ such that for $n\geq n_{a}$, the $\left(a,t\right)$-bound is bounded above by the $\left(0,t\right)$-bound.

To conclude the proof, we let
\begin{equation}
N\!\left(0,t\right)=\max\!\left\{n_{1},\dots,n_{a_{0}-1},2ta_{0}\right\}.
\end{equation}
\end{proof}

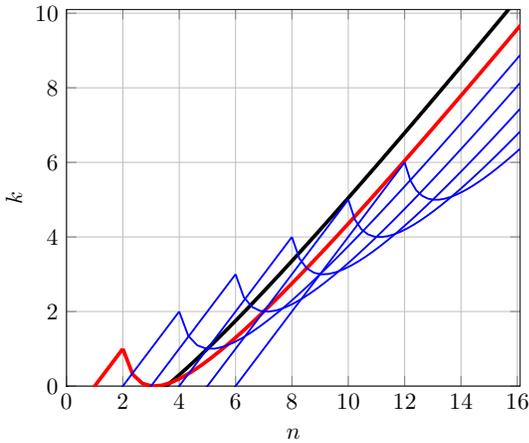
\begin{figure}[t]
\scalebox{0.88}{
\begin{tikzpicture}
\begin{axis}[
samples = 70,
xlabel = {$n$},
xmajorgrids=true,
xmax = 16.1,
xmin = 0.0,
y label style={at={(axis description cs:0.1,.5)},anchor=south},
ylabel = {$k$},
ymajorgrids=true,
ymax = 10.1,
ymin = 0.0,
domain = 0.0:26.1
]
\addplot[solid, black, ultra thick] {x - (ln(1 + 3 * x) / ln(2))};
\addplot[solid, red, ultra thick][domain=2.0:26.1] {x - 1 - (ln(1 + 3 * (x - 2)) / ln(2))};
\addplot[solid, red, ultra thick][domain=0.0:2.0] {x - 1};
\addplot[solid, blue, thick][domain=4.0:26.1] {x - 2 - (ln(1 + 3 * (x - 4)) / ln(2))};
\addplot[solid, blue, thick][domain=0.0:4.0] {x - 2};
\addplot[solid, blue, thick][domain=6.0:26.1] {x - 3 - (ln(1 + 3 * (x - 6)) / ln(2))};
\addplot[solid, blue, thick][domain=0.0:6.0] {x - 3};
\addplot[solid, blue, thick][domain=8.0:26.1] {x - 4 - (ln(1 + 3 * (x - 8)) / ln(2))};
\addplot[solid, blue, thick][domain=0.0:8.0] {x - 4};
\addplot[solid, blue, thick][domain=10.0:26.1] {x - 5 - (ln(1 + 3 * (x - 10)) / ln(2))};
\addplot[solid, blue, thick][domain=0.0:10.0] {x - 5};
\addplot[solid, blue, thick][domain=12.0:26.1] {x - 6 - (ln(1 + 3 * (x - 12)) / ln(2))};
\addplot[solid, blue, thick][domain=0.0:12.0] {x - 6};

\end{axis}
\end{tikzpicture}
}
\caption{The $\left(0,1\right)$-bound (the quantum Hamming bound for $t=1$) in black and the $\left(1,1\right)$-bound in red, along with $\left(\ell,1\right)$-bounds for $\ell=2,\dots,6$ in blue. The $\left(5,1\right)$-bound is the first bound with its local maxima below the quantum Hamming bound at $n=10$, so $N\!\left(0,t\right)=10$, since all $\left(\ell,1\right)$-bounds with $\ell<5$ have their intersection points at smaller vales of $n$. Other $N\!\left(\ell,t\right)$ values are shifts right by $2t\ell$.}
\label{shiftedfig}
\end{figure}

\cref{shiftedfig} shows how several $\left(\ell,1\right)$-bounds compare to each other. In this case where $t=1$, we can see that $n_{1}<\dots<n_{4}<2ta_{0}$, where $a_{0}=5$. We conjecture that this is always the case, that is 
\begin{equation}
N\!\left(0,t\right)=2ta_{0}.
\end{equation}

\begin{corollary}\label{mainresult}
There exists an $N\!\left(t\right)\in\mathbb{N}$ such that for all $n\geq N\!\left(t\right)$, any degenerate $\left[\!\left[n,k,2t+1\right]\!\right]$ stabilizer code satisfies the following inequality:
\begin{equation}\label{degenbound}
k\leq n-1-\log_{2}\!\left(\sum\limits_{i=0}^{t}3^{i}\binom{n-2t}{i}\right).
\end{equation}
\end{corollary}

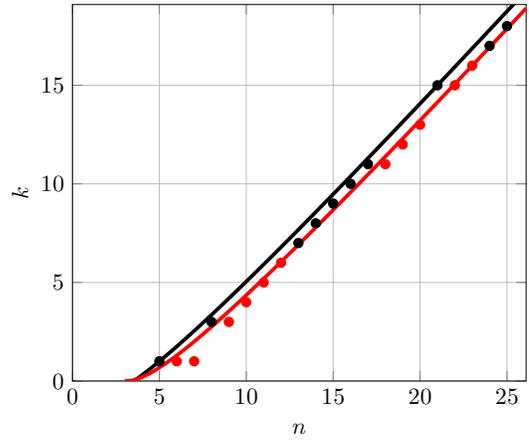
\begin{figure}[t]
\scalebox{0.88}{
\begin{tikzpicture}
\begin{axis}[
samples = 50,
xlabel = {$n$},
xmajorgrids=true,
xmax = 26.1,
xmin = 0.0,
y label style={at={(axis description cs:0.1,.5)},anchor=south},
ylabel = {$k$},
ymajorgrids=true,
ymax = 19.1,
ymin = 0.0,
domain = 0.0:26.1
]
\addplot[solid, black, ultra thick] {x - (ln(1 + 3 * x) / ln(2))};
\addplot[solid, red, ultra thick][domain=3.0:26.1] {x - 1 - (ln(1 + 3 * (x - 2)) / ln(2))};
\addplot[mark=*] coordinates {(5,1)};
\addplot[mark=*,red] coordinates {(6,1)};
\addplot[mark=*,red] coordinates {(7,1)};
\addplot[mark=*] coordinates {(8,3)};
\addplot[mark=*,red] coordinates {(9,3)};
\addplot[mark=*,red] coordinates {(10,4)};
\addplot[mark=*,red] coordinates {(11,5)};
\addplot[mark=*,red] coordinates {(12,6)};
\addplot[mark=*] coordinates {(13,7)};
\addplot[mark=*] coordinates {(14,8)};
\addplot[mark=*] coordinates {(15,9)};
\addplot[mark=*] coordinates {(16,10)};
\addplot[mark=*] coordinates {(17,11)};
\addplot[mark=*,red] coordinates {(18,11)};
\addplot[mark=*,red] coordinates {(19,12)};
\addplot[mark=*,red] coordinates {(20,13)};
\addplot[mark=*] coordinates {(21,15)};
\addplot[mark=*,red] coordinates {(22,15)};
\addplot[mark=*,red] coordinates {(23,16)};
\addplot[mark=*] coordinates {(24,17)};
\addplot[mark=*] coordinates {(25,18)};
\end{axis}
\end{tikzpicture}
}
\caption{The quantum Hamming bound for $t=1$ (black) and the bound for degenerate stabilizer codes from \cref{cordist3} (red). Optimal code parameters that must belong to nondegenerate stabilizer codes are shown as black points, while those that are allowed to belong to degenerate stabilizer codes, described by \cref{cordegenallow}, are shown as red points.}
\label{pointdegenboundfig}
\end{figure}

\begin{corollary}\label{cordist3}
Any degenerate $\left[\!\left[n,k,3\right]\!\right]$ stabilizer code satisfies:
\begin{equation}
k\leq n-1-\log_{2}\!\left(1+3\left(n-2\right)\right).
\end{equation}
\end{corollary}
\begin{proof}
This follows from the previous corollary where we can calculate that $N\!\left(1\right)\leq 12$, and by checking all degenerate distance-3 codes of length less than 12, we find that the result holds.
\end{proof}

Optimal distance-3 nondegenerate stabilizer codes have been completely described by Yu et al. \cite{Yu2013}, who showed that with a few exceptions most of these optimal codes saturate the quantum Hamming bound. Using this result, we can show precisely for which values of $n$ distance-3 degenerate stabilizer codes are not prohibited from being optimal:

\begin{corollary}\label{cordegenallow}
A degenerate stabilizer code is only able to encode the optimal number of logical qubits for the following lengths:
\begin{enumerate}
\item $n=8f_{m}+\left\{\pm1,2,3,4\right\}$ or
\item $n=f_{m+2}-\left\{\pm1,\pm2,3\right\}$,
\end{enumerate}
for some integer $m\geq1$ and $f_{m}=\left(4^{m}-1\right)/3$.
\end{corollary}

\cref{pointdegenboundfig} shows for $n\leq 25$ which optimal parameters are necessarily nondegenerate and which may potentially allow for a degenerate stabilizer code. As $n$ grows larger, the proportion of optimal stabilizer codes that are allowed to be degenerate grows smaller. Some degenerate stabilizer codes with these lengths are already known, for example those extending the perfect codes \cite{Li2013} or the quantum Hamming codes \cite{Gottesman1996} by one physical qubit using the construction in \cite[Theorem 6]{Calderbank1998}, while others are given as part of a construction of quantum-classical hybrid codes \cite{Nemec2021}. However, it remains an open question whether or not degenerate stabilizer codes encoding the optimal number of logical qubits for all of these lengths.

\section{Conclusion}

In this paper we introduce a new Hamming-like bound that applies to degenerate stabilizer codes, but not to nondegenerate stabilizer codes, after some positive integer $N\!\left(t\right)$. This new bound provides further evidence that degenerate stabilizer codes cannot violate the quantum Hamming bound, and suggests that the more degenerate (in terms of number of independent generators of weight less than $d$) a stabilizer code is, the further it can be bounded away from the quantum Hamming bound.

In the case of single-error-correcting stabilizer codes, a large number of both stabilizer and nonadditive codes saturate the quantum Hamming bound \cite{Gottesman1996, Grassl2009, Yu2013, Yu2015}, and some perfect codes achieve it with equality \cite{Gottesman1997, Li2013}. In this case, we proved that most of these optimal codes (in the sense that they saturate the linear programming bounds) cannot be degenerate. However, for codes with higher distances, even though the quantum Hamming bound seems to be as strict as the linear programming bounds \cite{Shor1997, Rains1998, Ashikhmin1999}, few if any known codes seem to get close to it \cite{GrasslONLINE}. In these cases it is hard to know if degenerate stabilizer codes must be bounded away from the optimal stabilizer codes as is the case when $t=1$.

\begin{table}[t]
    \centering
    \begin{tabular}{c|c|c|c}
        $t$ & Rains Bound & $M\!\left(t\right)$ & $N\!\left(t\right)$ \\ \hline
        1 & 5 & 5 & 12 \\
        2 & 11 & 9 & 60 \\
        3 & 17 & 14 & 150 \\
        4 & 23 & 20 & 288 \\
        5 & 29 & 25 & 470 \\
        6 & 35 & 30 & 696 \\
        7 & 41 & 35 & 980 \\
    \end{tabular}
    \caption{Comparison of $N\!\left(t\right)$ values from \cref{mainresult} for the bound for degenerate stabilizer codes and the analogous values $M\!\left(t\right)$ from \cite{Li2010} for the quantum Hamming bound for general quantum codes. Also shown is the bound $n\geq3d-4$ due to Rains \cite[Theorem 15]{Rains1999b} that was used in combination with $M\!\left(t\right)$ in \cite{Dallas2022} to prove the quantum Hamming bound for all $t<63$.}
    \label{tab:Nvaluetable}
\end{table}

We suspect that our new bound in \cref{degenbound} should generalize directly to the nonbinary case. A more interesting question is whether or not a similar bound can be found for nonadditive codes. In the case when $t=1$, there are multiple examples of nonadditive codes that get close to the quantum Hamming bound \cite{Grassl2009, Yu2015}, none of which appear to be degenerate, while alternatively there exist non-additive $\left(\!\left(7,2,3\right)\!\right)$ codes that are both impure and degenerate \cite{Pollatsek2004, Kubischta2023}, while some are impure but nondegenerate \cite{Cao2022}. It will be interesting to see whether the nondegenerate/degenerate or the pure/impure dichotomy is the correct choice for generalizations to general quantum codes.

Another open question is whether the $N\!\left(t\right)$ values from \cref{mainresult} can be improved upon. As shown in \cref{tab:Nvaluetable}, there is a significant gap between these values and the analogous $M\!\left(t\right)$ values for the quantum Hamming bound given by Li and Xing in \cite{Li2010}. While it is possible to obtain upper bounds on these values using techniques similar to ours, there the authors were able to get much lower values by clever manipulation of the linear programming bounds on the weight enumerators of codes. By using these values and a bound due to Rains \cite{Rains1999b}, Dallas et al. were able to show that the quantum Hamming bound holds for all codes with $t<63$. It would be interesting to see if the approach used by Li and Xing could lead to smaller values of $N\!\left(t\right)$ that could be used to prove \cref{mainresult} holds for all $n$ for values of $t>1$.

\section*{Acknowledgements}
The authors would like to thank Abhinav Anand and Balint Pato for fruitful discussions.
This work was supported by the Office of the Director of National Intelligence - Intelligence Advanced Research Projects Activity through an Army Research Office contract (W911NF-16-1-0082), the Army Research Office (W911NF-21-1-0005), and the National Science Foundation Institute for Robust Quantum Simulation (QLCI grant OMA-2120757). Support is also acknowledged from the Quantum Systems Accelerator.


\clearpage

\appendix

\onecolumngrid

\section*{Appendix} \label{appendix}
We look at the properties of the following function and its translations:

\begin{equation}
h_{t}\!\left(x\right)=
\begin{cases}
x-\log_{2}\!\left(\sum\limits_{i=0}^{t}3^{i}\binom{x}{i}\right) & \text{if } x\geq 0 \\
x & \text {if } x<0,
\end{cases}
\end{equation}
which when restricted to positive integers is the quantum Hamming bound for a $t$-error-correcting quantum code of length $n$. For simplicity, we refer to the polynomial inside the $\log$ function as $f_{t}\!\left(x\right)$. Real roots of $f_{t}$ will cause $h_{t}$ to be undefined, so we want to bound where these occur. We use the variable $n$ instead of $x$ when we want to restrict the function to the integers.

\begin{lemma}\label{binomiallemma}
For $0\leq n \leq t$, we have $h_{t}\!\left(n\right)=-n$.
\end{lemma}
\begin{proof}
It follows from the binomial theorem and the expansion of $4^{n}=\left(1+3\right)^{n}$ at these values.
\end{proof}

\begin{lemma}\label{rootboundlemma}
The function $h_{t}\!\left(x\right)$ is continuous for all $x\geq t-1$
\end{lemma}
\begin{proof}
It is easy to see that $3^{i}\binom{x}{i}>0$ for all $x>i-1$, as it is the product of positive terms. Since it is sum of these positive terms, $f_{t}\!\left(x\right)>0$ also holds for all $x>t-1$. Additionally, we also have $f_{t}\!\left(t-1\right)>0$, since the only nonpositive term is equal to zero. Therefore $h_{t}\!\left(x\right)$ is defined for all $x\geq t-1$ and is continuous on those values.
\end{proof}

\begin{lemma}\label{shift1-1}
For all $n\in\mathbb{N}$, we have
\begin{equation}
h_{t}\!\left(n\right)\leq h_{t}\!\left(n-1\right)+1.
\end{equation}
\end{lemma}
\begin{proof}
For $n\leq 0$, the two functions have the same value. At the point $n=1$, $-1=h_{t}\!\left(n\right)\leq h_{t}\!\left(n-1\right)+1=1$ by \cref{binomiallemma}. For $n>1$, this is equivalent to showing that
\begin{equation}
\sum\limits_{i=0}^{t}3^{i}\binom{n-1}{i}\leq\sum\limits_{i=0}^{t}3^{i}\binom{n}{i},
\end{equation}
which is true because
\begin{equation}
\binom{n-1}{i}\leq\binom{n}{i}
\end{equation}
holds for all $i\geq0$.
\end{proof}

We would like to identify the general structure of $h_{t}\!\left(n\right)$ so that we can identify how it interacts with its shifts in \cref{shiftprop}. In \cref{hminmax}, we show that $h_{t}\!\left(n\right)$ has two local optima, a local maxima at $n=0$ and a local minima between $n=2t-2$ and $n=2t$. In order to prove this, we must prove several lemmas first.

The derivative of $\binom{x}{i}$ is given by \begin{equation}\frac{d}{dx}\binom{x}{i}=\binom{x}{i}\sum\limits_{k=0}^{i-1}\frac{1}{x-k}.\end{equation}

\begin{lemma}\label{firstderivlowerlemma}
For all $t\in\mathbb{N}$, we have $h'_{t}\!\left(2t-2\right)<0$.
\end{lemma}
\begin{proof}
Since we have that
\begin{equation}\label{firstderivative}
h'_{t}\!\left(x\right)=1-\frac{f'_{t}\!\left(x\right)}{\ln2\cdot f_{t}\!\left(x\right)},
\end{equation}
we can see that $h'_{t}\!\left(x\right)<0$ if and only if $\ln2\cdot f_{t}\!\left(x\right)<f'_{t}\!\left(x\right)$. Therefore we will prove that
\begin{equation}\label{targetineq}
\ln2\sum\limits_{i=0}^{t}3^{i}\binom{2t-2}{i}<\sum\limits_{i=0}^{t}3^{i}\binom{2t-2}{i}\sum\limits_{k=0}^{i-1}\frac{1}{2t-2-k}.
\end{equation}

Note that we can bound $\ln2$ above and below by the following sums:
\begin{equation}
\sum_{k=0}^{t-2} \frac{1}{2t-2-k} < \ln2 = \int_{0}^{t-1}\frac{dx}{2t-2-x} < \sum_{k=1}^{t-1} \frac{1}{2t-2-k},
\end{equation}
or equivalently,
\begin{equation}\label{ln2bound}
\sum_{k=0}^{t-2} \frac{1}{2t-2-k} < \ln2 < \sum_{k=0}^{t-2} \frac{1}{2t-2-k}+\frac{1}{2(t-1)}.
\end{equation}

The following inequality is equivalent to our target inequality in \cref{targetineq}:
\begin{equation}\label{bothsideposineq}
\sum_{i=1}^{t-1} 3^i{2t-2 \choose i}\left(\ln 2 -\sum_{k=0}^{i-1}\frac{1}{2t-2-k}\right) < 3^t{2t-2 \choose t}\left(\sum_{k=0}^{t-1}\frac{1}{2t-2-k}-\ln 2\right),
\end{equation}
where all terms in the sum on the left-hand side are positive. Using the fact that
\begin{equation}
	\frac{{2t-2 \choose t-j}}{{2t-2 \choose t}} = \frac{(2t-2)!}{(t-j)!(t-2+j)!}\frac{t!(t-2)!}{(2t-2)!} = \frac{t(t-1)\cdots(t-j+1)}{(t-1)t\cdots(t-2+j)},
\end{equation}
\cref{bothsideposineq} is equivalent to
\begin{align}
    \left(\sum_{k=0}^{t-1}\frac{1}{2t-2-k}-\ln 2\right) >& \frac{1}{3}\frac{t}{t-1}\left(\ln 2 -\sum_{k=0}^{t-2}\frac{1}{2t-2-k}\right)+\frac{1}{3^2}\left(\ln 2 -\sum_{k=0}^{t-3}\frac{1}{2t-2-k}\right) \nonumber \\
    &+\frac{1}{3^3}\frac{t-2}{t+1}\left(\ln 2 -\sum_{k=0}^{t-4}\frac{1}{2t-2-k}\right)+\dots, 
\end{align}
which is equivalent to
\begin{align}\label{expandedsum}
    \ln 2\left(1+\frac{1}{3}\frac{t}{t-1}+\frac{1}{3^2}+\frac{1}{3^3}\frac{t-2}{t+1}+\dots\right) <& \sum_{k=0}^{t-1}\frac{1}{2t-2-k} +\frac{1}{3}\frac{t}{t-1}\sum_{k=0}^{t-2}\frac{1}{2t-2-k}+\frac{1}{3^2}\sum_{k=0}^{t-3}\frac{1}{2t-2-k} \nonumber\\
    &+\frac{1}{3^3}\frac{t-2}{t+1}\sum_{k=0}^{t-4}\frac{1}{2t-2-k}+\dots
\end{align}

We define the following terms:
\begin{align}
        z_0&=1,\;z_1=\frac{1}{3}\frac{t}{t-1},\;z_2=\frac{1}{3^2},\;z_3=\frac{1}{3^3}\frac{t-2}{t+1},\;\dots,\;z_j=\frac{1}{3^j}\frac{{2t-2 \choose t-j}}{{2t-2 \choose t}}, \nonumber \\
	Z&=z_0+z_1+\dots+z_{t-1}, \nonumber\\
	F&=\left(z_1+\dots+z_{t-1}\right)\frac{1}{t-1}+\left(z_2+\dots+z_{t-1}\right)\frac{1}{t}+\dots+z_{t-1}\frac{1}{2t-3}. \nonumber
\end{align}
Using this new notation, we can rewrite \cref{expandedsum} as
\begin{align}
	 Z\ln 2 &< z_0 \sum_{k=0}^{t-1}\frac{1}{2t-2-k} +z_1\sum_{k=0}^{t-2}\frac{1}{2t-2-k} + z_2\sum_{k=0}^{t-3}\frac{1}{2t-2-k} + z_3\sum_{k=0}^{t-4}\frac{1}{2t-2-k}+\dots  \\
  &=z_0\frac{1}{t-1}+(z_0+z_1)\frac{1}{t}+(z_0+z_1+z_2)\frac{1}{t+1}+\dots+(z_0+\dots+z_{t-1})\frac{1}{2t-2}  \\
  &=Z\left(\frac{1}{t-1}+\frac{1}{t}+\dots+\frac{1}{2t-2}\right)-F,
\end{align}
which is equivalent to
\begin{equation} 
 \ln 2 < \frac{1}{2t-2}+\frac{1}{2t-1}+\dots+\frac{1}{t}+\left(\frac{1}{t-1}-\frac{F}{Z}\right) = \sum_{k=0}^{t-2} \frac{1}{2t-2-k}+\left(\frac{1}{t-1}-\frac{F}{Z}\right).
\end{equation}

Note that
\begin{equation}\label{subtarget}
\left(\frac{1}{t-1}-\frac{F}{Z}\right) > \frac{1}{2(t-1)}\;\;\Leftrightarrow\;\;2(t-1)F-Z+1 < 1.
\end{equation}
So by the upper bound of $\ln 2$ given in \cref{ln2bound}, proving the right-hand side of the statement in \cref{subtarget} is sufficient for proving our target inequality \cref{targetineq}.

To prove the right-hand side of the statement in \cref{subtarget} holds, we expand it out as
\begin{align}
	&2(t-1)F-Z+1 \\
= & 2\left(z_1+\dots+z_{t-1}\right)+2\left(z_2+\dots+z_{t-1}\right)\frac{t-1}{t}+\dots+2z_{t-1}\frac{t-1}{2t-3}-\left(1+z_1+\dots+z_{t-1}\right)+1 \\
= & z_1+\left(1+2\frac{t-1}{t}\right)z_2+\left(1+2\frac{t-1}{t}+2\frac{t-1}{t+1}\right)z_3+...+\left(1+2\frac{t-1}{t}\dots+2\frac{t-1}{2t-3}\right)z_{t-1} \\
= & \left[\frac{t}{3(t-1)}+\frac{1}{9}+\frac{2(t-1)}{9t}\right]+\left(1+\frac{2(t-1)}{t}+\frac{2(t-1)}{t+1}\right)z_3+...+\left(1+2\frac{t-1}{t}\dots+2\frac{t-1}{2t-3}\right)z_{t-1} \\
= & \left[\frac{1}{3(t-1)}-\frac{2}{9t}+\frac{2}{3}\right]+\left(1+\frac{2(t-1)}{t}+\frac{2(t-1)}{t+1}\right)z_3+...+\left(1+2\frac{t-1}{t}\dots+2\frac{t-1}{2t-3}\right)z_{t-1} \\
< & \left[\frac{1}{3(t-1)}-\frac{2}{9t}+\frac{2}{3}\right]+(1+2r+2r)\frac{1}{3^3}+...+(1+2r(t-2))\frac{1}{3^{t-1}} \\
= & \left[\frac{1}{3(t-1)}-\frac{2}{9t}+\frac{2}{3}\right]+\sum_{k=3}^{t-1}\frac{1+2r(k-1)}{3^k},
\end{align}
where the inequality uses the facts that $z_p < \frac{1}{3^p}$ when $p\geq 3$ and $\frac{(t-1)}{t+q} < \frac{(t-1)}{t}$ when $q > 1$. We can then bound this sum by
\begin{align}
	\sum_{k=3}^{t-1}\frac{1+2r(k-1)}{3^k} &< \sum_{k=3}^{\infty}\frac{1+2r(k-1)}{3^k} \\
 &=(1-2r)\sum_{k=3}^{\infty}\frac{1}{3^k}+2r\sum_{k=3}^{\infty}\frac{k}{3^k} \\
 &=(1-2r)\left(\frac{1}{18}\right)+2r\left(\frac{7}{36}\right) \\
 &=\frac{1+5r}{18} \\
 &=\frac{1}{18}+\frac{5}{18}\left(\frac{t-1}{t}\right).
\end{align}
Recombining, we get
\begin{align}
2(t-1)F-Z+1&<\frac{1}{3(t-1)}-\frac{2}{9t}+\frac{2}{3}+\frac{1}{18}+\frac{5}{18}\left(\frac{t-1}{t}\right) \\
 &=1-\frac{t-3}{6t(t-1)} \\
 &\leq 1,
\end{align}
when $t\geq 3$, which is sufficient to show that our desired result is true when $t\geq 3$. We can check the cases $t=1,2$ numerically and see that $h'_{1}\!\left(0\right)=-3.328$ (from the positive direction) and $h'_{2}\!\left(2\right)=-0.488$, proving our desired result.

\end{proof}

\begin{lemma}\label{firstderivupperlemma}
We have $h'_{t}\!\left(2t\right)>0$ for all $t\in\mathbb{N}$.
\end{lemma}
\begin{proof}
By \cref{firstderivative}, we can see that $h'_{t}\!\left(x\right)>0$ if and only if $f'_{t}\!\left(x\right)<\ln2\cdot f_{t}\!\left(x\right)$.


To show $h'_{t}\!\left(2t\right)>0$, we have
\begin{align}
f'_{t}\!\left(2t\right) & = \sum\limits_{i=0}^{t}3^{i}\binom{2t}{i}\sum\limits_{k=0}^{i-1}\frac{1}{2t-k} \\
 & \leq \sum\limits_{i=0}^{t}3^{i}\binom{2t}{i}\sum\limits_{k=0}^{t-1}\frac{1}{2t-k} \\
 & = \sum\limits_{i=0}^{t}3^{i}\binom{2t}{i}\sum\limits_{k=1}^{t}\frac{1}{t+k} \\
 & < \sum\limits_{i=0}^{t}3^{i}\binom{2t}{i}\int\limits_{t}^{2t}\frac{dx}{x} \\
 & = \ln2\cdot\sum\limits_{i=0}^{t}3^{i}\binom{2t}{i}\\
 & = \ln2\cdot f_{t}\!\left(2t\right).
\end{align}
\end{proof}

\begin{lemma}\label{secondderivativelemma}
For all $x\geq t$, we have $h_{t}''\!\left(x\right)> 0$.
\end{lemma}
\begin{proof}
When restricted to $x\geq 0$, we have
\begin{equation}
h_{t}''\!\left(x\right)=\frac{\left(f_{t}'\!\left(x\right)\right)^2-f_{t}\!\left(x\right)f_{t}''\!\left(x\right)}{\ln2\left(f_{t}\!\left(x\right)\right)^{2}},
\end{equation}
where 
\begin{equation}
f_{t}'\!\left(x\right)=\sum\limits_{i=0}^{t}3^{i}\binom{x}{i}\sum\limits_{k=0}^{i-1}\frac{1}{x-k},
\end{equation}
and 
\begin{align}
f_{t}''\!\left(x\right) & = \sum\limits_{i=0}^{t}3^{i}\sum\limits_{k=0}^{i-1}\frac{d}{dx}\left(\binom{x}{i}\frac{1}{x-k}\right) \\
 & = \sum\limits_{i=0}^{t}3^{i}\binom{x}{i}\sum\limits_{k=0}^{i-1}\frac{1}{x-k}\left(\sum\limits_{j=0}^{i-1}\frac{1}{x-j}-\frac{1}{x-k}\right).
\end{align}

Note that by \cref{rootboundlemma}, $h_{t}''\!\left(x\right)$ will have no discontinuities for $x\geq t$. The following inequalities are equivalent:
\begin{align}
h_{t}''\!\left(x\right) & > 0 \\
\left(f_{t}'\!\left(x\right)\right)^{2} & > f_{t}\!\left(x\right)f_{t}''\!\left(x\right) \\
\sum\limits_{i=0}^{t}3^{i}\binom{x}{i}\sum\limits_{k=0}^{i-1}\frac{1}{x-k}\sum\limits_{r=0}^{t}3^{r}\binom{x}{r}\sum\limits_{s=0}^{r-1}\frac{1}{x-s} & > \sum\limits_{r=0}^{t}3^{r}\binom{x}{r} \sum\limits_{i=0}^{t}3^{i}\binom{x}{i}\sum\limits_{k=0}^{i-1}\frac{1}{x-k}\left(\sum\limits_{j=0}^{i-1}\frac{1}{x-j}-\frac{1}{x-k}\right) \\
\sum\limits_{i,r=0}^{t}3^{i+r}\binom{x}{i}\binom{x}{r}\left(\sum\limits_{k=0}^{i-1}\frac{1}{x-k}\sum\limits_{s=0}^{r-1}\frac{1}{x-s}\right) & > \sum\limits_{i,r=0}^{t}3^{i+r}\binom{x}{i}\binom{x}{r}\left(\sum\limits_{k=0}^{i-1}\frac{1}{x-k}\left(\sum_{s=0}^{i-1}\frac{1}{x-s}-\frac{1}{x-k}\right)\right).
\end{align}

It therefore suffices to show that for each $\left(i,r\right)$ pair that for $x\geq t$ we have
\begin{equation}
\sum\limits_{k=0}^{i-1}\frac{1}{x-k}\sum\limits_{s=0}^{r-1}\frac{1}{x-s} > \sum\limits_{k=0}^{i-1}\frac{1}{x-k}\left(\sum_{s=0}^{i-1}\frac{1}{x-s}-\frac{1}{x-k}\right)
\end{equation}

Again, it suffices to show that for each $\left(k,s\right)$ pair that we have 
\begin{equation}
\frac{1}{\left(x-k\right)\left(x-s\right)}>\frac{1}{\left(x-k\right)\left(x-s\right)}-\frac{1}{\left(x-k\right)^2},
\end{equation}
which is clearly true for $x\geq t$ as $k,s\leq t-1$. Therefore, we have that $h_{t}''\!\left(x\right)>0$ for all $x\geq t$.
\end{proof}

Combining these lemmas, we can describe the behavior of $h_{t}\!\left(n\right)$.

\begin{proposition}\label{hminmax}
The function $h_{t}\!\left(n\right)$ has two local optima, a local maxima at $n=0$ and a local minima between $n=2t-2$ and $n=2t$.
\end{proposition}
\begin{proof}
We prove this by showing that $h_{t}\!\left(n\right)$ is increasing on the interval $(-\infty, 0]$, decreasing on $[0, 2t-2]$, and increasing on $[2t, \infty)$. The first of these is trivial to show. By \cref{binomiallemma} it is clear that our function is decreasing on the interval $[0,t]$, so it must have a local maxima at $n=0$.

By \cref{rootboundlemma,secondderivativelemma}, we know that $h_{t}\!\left(x\right)$ is continuous and convex in the region $[t,\infty)$, and since $h'_{t}\!\left(2t-2\right)<0<h'_{t}\!\left(2t\right)$ by \cref{firstderivlowerlemma,firstderivupperlemma}, we know that it must have a single local minima between those two points.



\end{proof}

\begin{lemma}\label{lessonederivative}
We have $h'_{t}\!\left(x\right)<1$ for all $x\geq t$.
\end{lemma}
\begin{proof}
From \cref{firstderivative}, we know that this is equivalent to 
\begin{equation}
\frac{f'_{t}\!\left(x\right)}{\ln2\cdot f_{t}\!\left(x\right)}>0,
\end{equation}
for all $x\geq t$. By \cref{rootboundlemma}, we know that both the numerator and denominator must be positive in this range, so the result follows.
\end{proof}

\begin{proposition}\label{shiftprop}
Let $g_{t}\!\left(x\right)$ be the function representing the $\left(a,t\right)$-bound, that is
\begin{equation}
g_{t}\!\left(x\right)=h_{t}\!\left(x-2ta\right)+\left(2t-1\right)a.
\end{equation}
If $g_{t}\!\left(2ta\right)<h_{t}\!\left(2ta\right)$, then $g_{t}\!\left(n\right)<h_{t}\!\left(n\right)$ for all $n\in\mathbb{N}$.
\end{proposition}
\begin{proof}

We note that $g'_{t}\!\left(x\right)=1$ on the interval $[0,2ta]$. On the other hand, on the interval $[t,2ta]$ $h_{t}\!\left(x\right)$ is continuous by \cref{rootboundlemma} and $h'_{t}\!\left(x\right)<1$ by \cref{lessonederivative}, and on the interval $[0,t]$ the effective slope of $h_{t}\!\left(x\right)$ at integer points is $-1$ by \cref{binomiallemma}. Therefore, by our assumption that $g_{t}\!\left(2ta\right)<h_{t}\!\left(2ta\right)$, we have that $g_{t}\!\left(n\right)<h_{t}\!\left(n\right)$ for all $n\leq2ta$.

Let $x_{0}$ be the local minima of $g_{t}\!\left(x\right)$ described in \cref{hminmax}. Since $g_{t}\!\left(x\right)$ is decreasing at integer points on the interval $[2ta,x_{0}]$, and $h_{t}\!\left(x\right)$ is increasing on the same interval, it follows from our initial assumption that $g_{t}\!\left(n\right)<h_{t}\!\left(n\right)$ for all $n\in[2ta,x_{0}]$.

Note that $g'_{t}\!\left(x\right)=h'_{t}\!\left(x-2ta\right)$, so on the interval $[x_{0},\infty)$ we have $g'_{t}\!\left(x\right)<h'_{t}\!\left(x\right)$ since $h'\!\left(x\right)$ is increasing by \cref{secondderivativelemma}. Consider the function $h_{t}\!\left(x\right)-g_{t}\!\left(x\right)$ and its derivative $h'_{t}\!\left(x\right)-g'_{t}\!\left(x\right)$, which must be positive on the interval in question. Therefore, we have that $h_{t}\!\left(x\right)-g_{t}\!\left(x\right)>0$, or equivalently that $g_{t}\!\left(n\right)<h_{t}\!\left(n\right)$ for all $n\geq x_{0}$.

\end{proof}

\end{document}